\documentclass{article}
\usepackage[margin=1in]{geometry}

\usepackage[utf8]{inputenc} 
\usepackage[T1]{fontenc}
\usepackage{cite}             

\usepackage[cmex10]{amsmath}  
                       \usepackage{amsfonts}
\usepackage{amsthm, amssymb}
\usepackage{ comment}
\usepackage[draft,bookmarks=false]{hyperref}

\usepackage{bbm}
\usepackage{mathrsfs} 
\usepackage{tikz}
\usetikzlibrary{positioning, calc}
\usepackage{stmaryrd }

\usepackage{tikz}

\usetikzlibrary{fit}					
\usetikzlibrary{backgrounds}	

\newtheorem{theorem}{Theorem}
\newtheorem{lemma}[theorem]{Lemma}

\newtheorem{definition}[theorem]{Definition}

\newtheorem{remark}[theorem]{Remark}       
\usepackage{mleftright}       
\mleftright                   

\usepackage{graphicx}         
\usepackage{booktabs}         

\usepackage[ruled,vlined]{algorithm2e} 

\usepackage{algorithmic}

\DeclareMathOperator*{\argmin}{\mathbf{arg\,min}}

\title{Source Coding with Free Bits and the Multi-Way Number Partitioning Problem\\
}

\author{
\begin{tabular}{cc}
  \begin{tabular}[t]{c}
    \textbf{Niloufar Ahmadypour} \\
    \textit{Dept. of Electrical and Computer Eng.} \\
    \textit{Isfahan University of Technology} \\
    Isfahan, Iran \\
    \texttt{n.ahmadypour@iut.ac.ir}
  \end{tabular}
  &
  \begin{tabular}[t]{c}
    \textbf{Amin Gohari} \\
    \textit{Dept. of Information Eng.} \\
    \textit{The Chinese University of Hong Kong} \\
    Hong Kong, China \\
    \texttt{agohari@ie.cuhk.edu.hk}
  \end{tabular}
\end{tabular}
}

\begin{document}
\maketitle

\begin{abstract}
We introduce a new variant of variable-length source coding for sending a source over two parallel channels, one of which is costly and the other free. We give a complete solution to this problem. Next, we relate the problem to the number partitioning problem, which is the task of dividing a given list of numbers into a pre-specified number of subsets such that the sum of the numbers in each subset is as nearly equal as possible. We introduce two new objective functions for this problem and show that an adapted version of the Huffman coding algorithm (with a runtime of $\mathcal{O}(n \log n)$ for input size $n$) produces the optimal solution for one objective function, and a nearly optimal solution for the other objective function.
\end{abstract}

\section{Introduction}
We consider an information-theoretic source coding problem in which a discrete
memoryless source $X$ with alphabet $\mathcal{X}$ is available at the encoder;
see Fig.~\ref{fig2}. The encoder communicates with the decoder over two
noiseless parallel channels. The first channel, depicted as channel (A) in
Fig.~\ref{fig2}, is a free channel and conveys a symbol
$A \in \{1,2,\ldots,k\}$. The second channel is costly, and the encoder is
charged for each transmitted bit. The objective is to design an encoding
scheme that allows the decoder to perfectly and instantaneously reconstruct
the source $X$ while minimizing the average number of bits transmitted over
the costly channel.

When $k = 2^{k'}$, this problem reduces to finding a zero-error, prefix-free
binary coding scheme $\mathcal{C}$ that solves the following optimization
problem:
\begin{align}
\label{opteq}
\mathcal{C}^* = \argmin_{\mathcal{C}:\mathcal{X} \to \{0,1\}^*}
\mathbb{E}\left[\bigl( l_{\mathcal{C}}(X) - k' \bigr)_+\right] ,
\end{align}
where $l_{\mathcal{C}}(x)$ denotes the length of the codeword assigned to the
symbol $x$ under the code $\mathcal{C}$, and $(t)_+ := \max\{t,0\}$.

The optimal coding scheme admits a natural partition-based interpretation.
Specifically, the free channel conveys the index of a subset of source symbols,
while the costly channel carries a prefix-free binary codeword that uniquely
identifies the source symbol within the indicated subset.\footnote{It resembles superposition coding in broadcast channels, with the message on the free channel resembling the cloud center and the message on the costly channel resembling the satellites.} In this manner, the
encoding problem reduces to finding a partition of the source alphabet into
$k$ subsets and an associated collection of prefix-free codes that jointly
minimize the expected number of transmitted bits on the costly channel.

\begin{figure}[t]
    \centering
    \begin{tikzpicture}[
        scale=1, 
        transform shape,
        block/.style={
            draw, 
            rectangle, 
            minimum height=1.2cm, 
            minimum width=2.0cm, 
            align=center,
            font=\large
        },
        thick line/.style={
            thick,
            shorten >=1pt,
            shorten <=1pt
        }
    ]

        \node [block] (source) {Source X};
        \node [block, right=1.0cm of source] (encoder) {Encoder};
        \node [block, right=3.0cm of encoder] (decoder) {Decoder};
        \node [block, right=1.0cm of decoder] (receiver) {Receiver};

        \draw [thick line] (source) -- node[above] {x} (encoder);

        \draw [thick, dashed] ([yshift=2pt]encoder.east) -- ([yshift=2pt]decoder.west);
        \draw [thick] ([yshift=-2pt]encoder.east) -- ([yshift=-2pt]decoder.west);
        
        \node [above=0.2cm] at ($(encoder.east)!0.5!(decoder.west)$) {\large Channel (A)};

        \draw [thick line] (decoder) -- node[above] {$\hat{x}$} (receiver);

    \end{tikzpicture}
\caption{Transmission of a discrete memoryless source over two parallel
noiseless channels: a free channel conveying the subset index and a costly
binary channel conveying the remaining codeword bits.}    \label{fig2}
\end{figure}

When $k=1$, the problem reduces to the standard variable-length source coding problem, for which Huffman coding is optimal. We prove that a generalized version of Huffman coding, \emph{Early stopping Huffman algorithm,} is optimal for the problem with general $k$. 

\subsection{Number partitioning problem}
We relate our source coding problem to the number partitioning problem, which is also a focus of this work. Indeed, the ``early stopping Huffman algorithm'' is already introduced in \cite[Theorem~1]{rvh} as an \emph{approximate} solution to the number partitioning problem (for our source coding problem, it is the exact solution).

Let $S=(\beta_1, \beta_2, \cdots, \beta_n)$ be a list of $n$ positive numbers. The number partitioning problem is the task of partitioning $S$ into $k$~subsets $S_1, S_2, \cdots, S_k$ so that the sum of the numbers in different subsets ($q_i=\sum_{\beta_j \in S_i } \beta_j$ for $1 \leq i \leq k$) are as nearly equal as possible. For instance, if  $S=(1,1,2,3,4,5)$ and $k=2$, we can consider the following partition $(1,1,2,4)$ and $(3,5)$. The numbers in each subset add up to 8, so this is a completely balanced partition. The number partitioning problem finds various applications such as multi-processor scheduling (minimizing the total time of processing) \cite{garey79,ggreedy69} or voting manipulation \cite{walsh}; see also \cite{worst2007, korf98}.

Different objective functions for the number partition problem can be defined \cite{ethan18,rvh}:
\begin{enumerate}
\item \label{item1}[\emph{Min-Difference objective function}]
Minimize the difference between the largest and smallest subset sums, \emph{i.e.,} minimize $\max_{1\leq i\leq k}q_i- \min_{1\leq i\leq k}q_i$,
\item \label{item2}[\emph{Min-Max objective function}]
Minimize the largest subset sum, \emph{i.e.,} minimize $\max_{1\leq i\leq k}q_i$,
\item \label{item3}[\emph{Max-Min objective function}]
Maximize the smallest subset sum,  \emph{i.e.,} maximize $\min_{1\leq i\leq k}q_i$,
\item \label{item4}[\emph{Entropic objective function}]
Minimize the KL divergence distance between the normalized subset-sum vector 
$(q_1/M, q_2/M, \cdots, q_k/M)$ and the uniform distribution $(1/k, 1/k, \cdots, 1/k)$, where 
$M=\sum_{i=1}^k q_i=\sum_{i=1}^n \beta_i$ is sum of all numbers in the list and the KL divergence between two distributions $p_X$ and $q_X$ is defined as \begin{align}D(p_X \| q_X)= \sum_{x \in \mathcal{X}}p_X(x) \log \dfrac{p_X(x)}{q_X(x)}.
 \end{align}

Equivalently, we can maximize the entropy of the normalized subset-sum vector 
$$H(q_1/M, q_2/M, \cdots, q_k/M)=\sum_{i=1}^k \frac{q_i}{M}\log\frac{M}{q_i}$$
where all logarithms in this paper are in base two.
\end{enumerate}
While these objective functions are all equivalent when $k=2$, none of them are equivalent to the others for $k>2$ \cite{korf2010}. 
For the case of $k=2$, Karp proved that the decision version of the number partitioning problem is NP-complete \cite{karp72}. For the entropic objective function and for $k\geq 2$, the authors in \cite{rvh} prove that finding the function that maximizes the entropy is NP-hard. However, there are some algorithms such as pseudo-polynomial time dynamic programming solutions or some heuristic algorithms that solve the problem approximately or completely \cite{
 korf2009, korf2010, karp72, ethan18, korf2014, moffit2013,
ggreedy69,korf95,karp82, korf2011,rvh}. 
For instance, the greedy algorithm gives a $(4/3 -  1/(3k))$-approximation
 in time $\mathcal{O}(2^k n^2)$ for the Min-Max objective function \cite{ggreedy69}, and this algorithm is optimal for $n\leq k+2$ \cite{korf2011}. For the entropic objective function, in \cite[Theorem~$1$]{rvh}, the authors present the ``Early stopping Huffman algorithm,"  with the time complexity of $\mathcal{O}(n\log(n))$ which yields a partition whose entropy is at most $0.086...$ less than the maximum entropy value.

While the four objective functions are not equivalent to each other for $k>2$, they are not unrelated. The max–min inequality implies that the Min-Max objective function is always greater than or equal to the Max-Min objective function. Moreover, the Min-Max objective function is also related to the entropic objective function. To see this, note that the Min-Max objective function is equivalent to 
maximizing the min-entropy of the normalized subset-sum vector 
$(q_1/M, q_2/M, \cdots, q_k/M)$; here the 
min-entropy of a discrete random variable $A$ is defined as 
$H_\infty(A) = -\log \max_{a\in\mathcal{A}} P_A(a).$ Note that the min-entropy serves as a lower bound on the Shannon entropy of $A$ defined by $H(A)=-\sum_a p_A(a)\log(p_A(a))$. Thus, a partition that is optimal for the Min-Max objective (which maximizes min-entropy) is not necessarily optimal for the entropic objective (which maximizes Shannon entropy).

Both the min-entropy and the Shannon entropy are the limit cases of the R\'enyi entropy of order $\alpha \in (0,1) \cup (1,\infty)$ defined as follows:
\begin{align}
 H_{\alpha}(A)=\dfrac{1}{1-\alpha} \log \left( \sum_{a \in \mathcal{A}} p_A^\alpha (a) \right),   
 \end{align}
As $\alpha$ tends to one, the R\'enyi entropy reduces to the Shannon entropy, and as $\alpha$ tends to infinity, it reduces to the min-entropy. In fact, in \cite{sason}, the author studies the R\'enyi-entropic objective function and generalizes the results of \cite{rvh}.

\subsection{New objective functions for the number partition problem}
Using the source coding problem introduced earlier, in this paper, we define a 5th objective function, closely related to the entropic objective function, which we call the \emph{compression objective function}. We prove that the \emph{exact solution} of the optimization problem with the compression objective function can be found in time  $\mathcal{O}(n\log(n))$ using the Early stopping Huffman algorithm.

Next, as stated above, the Min-Max objective function can be understood in terms of the information-theoretic concepts (the min-entropy). What about the Max-Min objective function? We define another new class of objective functions called the \emph{$\alpha$-divergence objective function} using the $\alpha$-R\'enyi divergence that reduces to the Max-Min objective function as $\alpha$ tends to infinity. When $\alpha=1$, this objective function reduces to the \emph{product objective function} which maximizes the product of subset-sum values, i.e., maximizes 
\begin{equation}
    \label{prod-objf}
\prod_{i=1}^k q_i.\end{equation}
This measure is a meaningful measure of uniformity since $\sum_{i=1}^kq_i$ is fixed and $\prod_{i=1}^k q_i$ reaches its maximum when $q_i$'s are equal. 
    This objective function is known as the Nash collective utility function (CUF) in the literature on fair division and resource allocation. We also study bounds on the optimal solution of the new objective function.

\color{black}

This paper is organized as follows: in Section~\ref{s2}, we introduce the compression objective function and the source coding problem. We show that the variant of the Huffman algorithm \cite{rvh} yields the exact solution for this objective function. In Section~\ref{s3}, we relate this to the number partitioning problem and present the $\alpha$-divergence objective function, its relation with other ones, and the bounds on its optimal solution. Also, a principle of optimality is stated and proved for the entropy and the product objective functions.
\color{black} 

\section{Preliminaries}

\begin{definition}
    Suppose that $\boldsymbol{p}$ and $\boldsymbol{q}$ are two sequence in  $\mathbb{R}^n$ that are sorted in descending order and $\sum_{i=1}^n p_i =\sum_{i=1}^n q_i$, we say that $\boldsymbol{p}$ is majorized by $\boldsymbol{q}$ or equivalently $\boldsymbol{q}$ majorizes $\boldsymbol{p}$ (written as $\boldsymbol{p}\preceq \boldsymbol{q}$) if and only if $\sum_{k=1}^i p_k \leq \sum_{k=1}^i q_k$ for all $1 \leq i \leq n$. Also a function $f:\mathbb{R}^n \to \mathbb{R}$ is Schur convex (Schur concave) if for all $\boldsymbol{p},\boldsymbol{q} \in \mathbb{R}^n$, such that $\boldsymbol{p} \preceq \boldsymbol{q}$ one has that $f(\boldsymbol{p}) \leq f(\boldsymbol{q})$ ($f(\boldsymbol{p}) \geq f(\boldsymbol{q})$).
\end{definition}

\begin{definition}
For a convex function $f: (0, \infty) \to \mathbb{R}$, with $f(1)=0$ and two probability mass functions $p_{X}(x)$ and $q_{X}(x)$ defined on a finite set $\mathcal{X}$ the $f$-divergence is defined as follows
\begin{align}
   D_f(p_X \| q_X) = \sum_{x \in \mathcal{X}} q_X(x)f\left( \dfrac{p_X(x)}{q_X(x)} \right), 
\end{align}
where
\begin{align*}
    &f(0):=\lim_{x \to 0^+} f(x),\\
    &0f(\dfrac{0}{0}):=0,\\
    &0f(\dfrac{a}{0}):= \lim_{x \to 0^+}xf(\dfrac{a}{x}) =a\lim_{x \to \infty}\dfrac{f(x)}{x}, \; a>0.
    \end{align*}
\end{definition}
A well-known subclass of $f$-divergences is the class of Alpha-divergences. The parametric functions that generate this class are defined as follows. For $\alpha \in \mathbb{R}$
\begin{align}
 f_\alpha(t) := \begin{cases}
 \dfrac{t^\alpha-\alpha t-(1-\alpha)}{\alpha(\alpha-1)}, & \alpha \neq 0, \alpha \neq 1,\\
 t\ln t -t+1, &   \alpha=1, \\
 -\ln t+t-1, & \alpha=0.
 \end{cases}\label{deffalpha}
\end{align}
The function $f_\alpha:(0,\infty) \to \mathbb{R}$ is a convex non-negative function with $f_\alpha (1)=0$.

\begin{definition}[The R\'enyi divergence \cite{renyi}]
Suppose that $p_X$ and $q_X$ are probability mass functions defined on $\mathcal{X}$, the R\'enyi divergence of order $\alpha \in (0,1) \cup (1,\infty)$ is defined as
 \begin{align}
     D_{\alpha}(p_X \| q_X)=\dfrac{1}{\alpha-1} \log \sum_{x \in \mathcal{X}} p_X^{\alpha}(x)q_X^{1-\alpha}(x),
 \end{align}
 by continuity we have
 \begin{align}
   \label{dinf}
     D_{\infty}(p_X \| q_X) &= \log \sup_{x \in \mathcal{X}} \dfrac{p_X(x)}{q_X(x)}.
 \end{align}
 \label{defR}
\end{definition}
Using $f_\alpha$ as defined in \eqref{deffalpha}, for $\alpha\in~(0,\infty)$, one can directly check the following well-known one-to-one correspondence
\begin{align} \label{eq57}
  D_{\alpha} (\boldsymbol{p} \| \boldsymbol{q}) = \begin{cases} 
  \dfrac{1}{\alpha-1} \log \left(  1+\alpha (\alpha -1)  D_{f_\alpha} (\boldsymbol{p} \| \boldsymbol{q}) \right),    &   \alpha \neq 1, \\
  (\log e) D_{f_1} (\boldsymbol{p} \| \boldsymbol{q}), & \alpha =1.
    \end{cases}
\end{align}
The next lemma is a special case of \cite[Theorem~7]{sasonnew} which presents an upper bound used in the proof of Theorem~\ref{main1}.
\begin{lemma} \label{lemma-w}
Suppose that $\boldsymbol{P}_n(\rho)$ for $n\geq 2$ and $\rho > 1$ is the set of all probability mass vectors $\boldsymbol{p}$ (sorted in descending order) of size $n$ with $\dfrac{p_1}{p_n} \leq \rho$, and $\boldsymbol{u}_n$ is the Uniform probability mass vector, then for each $\boldsymbol{p} \in \boldsymbol{P}_n(\rho)$, if $\alpha \in (0,1) \cup (1,\infty)$ we have
\begin{align}\nonumber
    D_\alpha(\boldsymbol{u}_n \| \boldsymbol{p}) \leq& \dfrac{1}{\alpha-1}\log \left(  1+ \dfrac{1+\alpha(\rho-1)-\rho^{\alpha}}{(1-\alpha)(\rho-1)} \right)\\ \label{eq73}
    & - \dfrac{\alpha}{\alpha-1}\log \left(  1+ \dfrac{1+\alpha(\rho-1)-\rho^{\alpha}}{(1-\alpha)(\rho^{\alpha}-1)} \right),
\end{align}
and if $\alpha =1$ we get
\begin{align}\label{eq74}
    D_\alpha(\boldsymbol{u}_n \| \boldsymbol{p}) \leq& \dfrac{\rho \log \rho}{\rho-1} 
     - \log \left(   \dfrac{e\rho \ln \rho} {\rho-1} \right),
\end{align}
where the RHS of \eqref{eq74} is the limit $\alpha \to 1$ of the RHS of \eqref{eq73}.   
\end{lemma}
\begin{proof}
Using Theorem~7 of \cite{sasonnew}, for each $n\geq 2$, $\rho > 1$ and $\boldsymbol{p} \in \boldsymbol{P}_n(\rho)$,  we know that:
\begin{align}
 D_{f_\alpha} (\boldsymbol{u}_n \| \boldsymbol{p} )  &\leq  \lim_{n \to \infty}  \max_{\boldsymbol{q} \in \boldsymbol{P}_n(\rho)} D_{f_\alpha} (\boldsymbol{u}_n \| \boldsymbol{q} ).
\end{align}
Now Corollary~1 and Eq.~(149) of \cite{sasonnew} result in for all $\alpha\in~(0,1) \cup (1,\infty)$
\begin{align}\nonumber
 & \lim_{n \to \infty}  \max_{\boldsymbol{q} \in \boldsymbol{P}_n(\rho)} D_{f_\alpha} (\boldsymbol{u}_n \| \boldsymbol{q}) \\ \label{eq55}
  & \ \  = \dfrac{1}{\alpha (\alpha-1)} \left( \dfrac{(1-\alpha)^{\alpha-1}(\rho^\alpha-1)^\alpha(\rho-\rho^\alpha)^{1-\alpha}}{(\rho-1)\alpha^\alpha}-1 \right), 
  \end{align}
and for $\alpha =1$ by a continuous extension we get
\begin{align} \label{eq56}
  & \lim_{n \to \infty}  \max_{\boldsymbol{q} \in \boldsymbol{P}_n(\rho)} D_{f_1} (\boldsymbol{u}_n \| \boldsymbol{q}) = \dfrac{\rho \ln (\rho-1)}{\rho -1}- \ln \left( \dfrac{e \rho \ln \rho}{\rho-1} \right).   
\end{align}

Therefore, using Equations \eqref{eq55}, \eqref{eq56} and \eqref{eq57} for $\alpha\in~(0,1) \cup (1,\infty)$  we get \eqref{eq73} and for $\alpha  =1$ we get \eqref{eq74}. 
\end{proof}
\section{The Source Coding Problem}\label{s2}
Let $X\sim P_X$ be a given discrete random variable. Assume that $A=f(X)$ is transmitted on the free channel for some function $f:\mathcal{X}\rightarrow\mathcal{A}$ where $\mathcal{A}=\{1,2,\cdots, k\}$. Let $p_i(x)$ be the conditional distribution of $X$ given $A=i$ for $1\leq i\leq k$. Let
$\mathcal{C}_i$ be the Huffman code for compressing $X$ when $X\sim p_i(x)$. Then, the objective function in \eqref{opteq} equals
\begin{align}L(X|A)\triangleq\sum_{i}\mathbb{P}[A=i]\mathbb{E}(\ell_i(X)|A=i)\label{eqnLXA2}\end{align}
where $\ell_i(x)$ is the length of the Huffman code assigned to symbol $x$ in $\mathcal{C}_i$.

\begin{algorithm}[t]
 \vspace{1em}
 Input a list $S_0=(\beta_1, \beta_2, \cdots, \beta_n)$ and $k$\;
 Set $i\leftarrow 0$\;
 \While{$|S_i|>k$}{
  Sort the list in $S_i$ in increasing order as $S_i=(b_1,b_2,...,b_m)$ where
  $b_1\leq b_2\leq\cdots\leq b_m$
  \;
  Merge the smallest numbers $b_1$ and $b_2$ together and form the list
  $S_{i+1}=(b_1+b_2, b_3, b_4, \cdots, b_m)$
  \;
  Increase $i$ by one\;
 }
 Put all the numbers merged together in the same group. 

\caption{Early Stopping Huffman Algorithm}\label{algorithm}
\end{algorithm}
Consider Algorithm \ref{algorithm}.  This algorithm is similar to Huffman coding except that the algorithm is stopped prematurely when the size of the list becomes equal to $k$. We prove that Algorithm \ref{algorithm}, when applied to the sequence  $\{p_X(x), x\in\mathcal{X}\}$ gives the optimal solution for solving 
\begin{align}\label{eqnLXA}
\argmin_{ \small{f:\mathcal{X}\rightarrow \mathcal{A} }} L(X|A)
\end{align} 
where $A=f(X)$. 

Before we state the main theorem about optimality of the algorithm, we explain the procedure with an example. Let $k=2$ (one bit on the free channel).
 Consider the following distribution: $$X\sim(1/16,1/16,2/16,3/16,4/16,5/16).$$ The algorithm produces the following lists:
$(1/16,1/16,2/16,3/16,4/16,5/16)\mapsto (\textbf{2}/16,2/16,3/16,4/16,5/16)\mapsto (3/16,\textbf{4}/16,4/16,5/16)
\mapsto (4/16,5/16,\textbf{7}/16)
\mapsto (\textbf{7}/16,\textbf{9}/16)$. One can see that the numbers $1/16,1/16,2/16,3/16$ are grouped together (adding up to $7/16$), and $4/16,5/16$ are also grouped together during the execution of the algorithm (adding up to $9/16$). This yields the optimal partitioning of the input symbols into two groups  $(1/16,1/16,2/16,3/16)$ and $(4/16,5/16)$ and defines the optimal function $f$. 

\begin{theorem}
\label{mainthm} Let $S=(\beta_1, \beta_2, \cdots, \beta_n)$ be a list of probability weights. 
For $n>k \geq 2$,
Algorithm~\ref{algorithm} gives an optimal solution of \eqref{eqnLXA}.
\end{theorem}
Before we prove Theorem~\ref{mainthm}, we state the following lemma.
 \begin{lemma}
 \label{implemma}
Assume $n>k\geq 2$. There is an optimal mapping $f$ minimizing \eqref{eqnLXA}, such that the two smallest numbers are in the same partition, \emph{i.e.,} $f(1)=f(2)$ for a list  $S=(\beta_1, \beta_2, \cdots, \beta_n)$ of probability
weights where $\beta_1\leq \beta_2\leq \beta_3\leq \ldots \leq \beta_n$.
 \end{lemma}
\begin{proof}[Proof of Lemma~\ref{implemma}]
Take a list
$S=(\beta_1, \beta_2, \cdots, \beta_n)$ where $\beta_1\leq \beta_2\leq \ldots \leq\beta_n$. Let $f$ be an optimal mapping minimizing \eqref{eqnLXA} such that $f(1)\neq f(2)$. There are two cases: 
\begin{enumerate}
\item
\label{case1} One cannot find $i\in \{3,4,\cdots, n\}$ such that $f(i)=f(1)$ or $f(i)=f(2)$. In this case, the Huffman code for $X$ given $A=f(1)$ or $A=f(2)$ has zero length.  Since $n>k$, one can find numbers $j_1, j_2\in \{3,4,\cdots, n\}$ such that 
$f(j_1)=f(j_2)=a$ for some $a\notin \{f(1), f(2)\}$.
We construct a new partition function $f'(\cdot)$ such that $f'(j_1)=f(1)$, $f'(1)=a$ and $f'(j_2)=f(2)$, $f'(2)=a$ and $f(\cdot)$, $f'(\cdot)$ are equal on the other values, then the expected length $L(X|A)$ decreases by $$\Delta=(\beta_{j_1}-\beta_1)\ell_a(j_1)+(\beta_{j_2}-\beta_2)\ell_a(j_2),$$ where $\ell_a(j_1)$ and $\ell_a(j_2)$ are the lengths of the Huffman codewords assigned to $X=j_1$ and $X=j_2$ conditioned on $A=a$.\footnote{We are constructing a feasible (not necessarily Huffman-optimal) code for the new partition by reusing old codewords. The new Huffman optimum has a length less than or equal to that feasible code's length, giving the desired contradiction.} 
This is a contradiction with optimality of $f$ unless $\Delta=0$. If $\Delta=0$, $f'$ will also be  an optimal partition function satisfying $f'(1)=f'(2)$.

\item  
There exists some $i \in \{3,4,\cdots, n\}$ such that either $f(i)=f(1)$, or $f(i)=f(2)$. Let $a=f(1)$ and $b=f(2)$. Let $\mathcal{C}_{a}$ and $\mathcal{C}_{b}$ be the Huffman codes for the distribution of $X$ given $A=a$ and, $A=b$ respectively. At least one of the Huffman codes $\mathcal{C}_{a}$ and $\mathcal{C}_{b}$ has a non-zero average length. 
In any Huffman code with at least two symbols,  two of the longest codewords have the same length, and they are assigned to symbols with the lowest probabilities \cite{huffman}. First assume that $\ell_{a}(1) \geq  \ell_{b}(2)$. Then, $\mathcal{C}_{a}$ certainly has more than one codeword. Since $X=1$ has the least probability (corresponds to $\beta_1$), it has the least probability in its group, and also its codeword has the largest length in code $\mathcal{C}_{a}$. Moreover, there is another codeword with this length that corresponds to some $i_1\in \{3,4,\cdots, n\}$. Thus, $f(i_1)=a$ and 
$\ell_a(i_1)=\ell_a(1)$. Construct the new partition function $f'(\cdot)$ such that $f'(i_1)=f(2)$, $f'(2)=f(1)$ and $f(\cdot)$, $f'(\cdot)$ are equal on the other values. Using the same Huffman codewords as before (the codeword that was assigned to $2$ should now be assigned to $i_1$ and vice versa),  this change in the mapping reduces the expected length of codewords by at least
 $\Delta=(\ell_a(1)-\ell_b(2))(\beta_{i_1}-\beta_2).$ This is a contradiction unless $\Delta=0$ which implies optimality of $f'$. For the case $\ell_{a}(1) <  \ell_{b}(2)$, a similar argument goes through. 
 Therefore, similar to Case~\ref{case1}, we can construct an optimal mapping satisfying $f'(1)=f'(2)$.\qedhere
\end{enumerate}
\end{proof}
\begin{proof}[Proof of Theorem~\ref{mainthm}]

Consider a list $S=(\beta_1, \beta_2, \cdots, \beta_n)$ where $\beta_1\leq \beta_2\leq \cdots \leq \beta_n$. We show in Lemma~\ref{implemma} that there is an optimal partition (minimizing $L(X|A)$) such that the two smallest numbers in the list, namely $\beta_1$ and $\beta_2$, belong to the same subset in that partition. Knowing this, we can simply merge these two numbers together and replace $\beta_1$ and $\beta_2$ by $\beta_1+\beta_2$. 
We claim that the problem then reduces to finding an optimal partition for a new list
$(\beta_1+\beta_2, \beta_3,\cdots, \beta_n)$. The reason is as follows: assume that $\beta_1$ and $\beta_2$ are the weights of symbols $1$ and $2$. Assume that $f(1)=f(2)=i$ for some $1\leq i\leq k$. Then, in the distribution of $X$ conditioned on $A=i$, the probabilities $\beta_1/p(A=i)$ and $\beta_2/p(A=i)$ are still the two smallest numbers. It is known that a Huffman code starts off by merging the two symbols of lowest probabilities. Therefore, as $\beta_1/p(A=i)$ and $\beta_2/p(A=i)$ are in the same group, an optimal Huffman code also begins by merging $\beta_1/p(A=i)$ and $\beta_2/p(A=i)$ into a symbol $(\beta_1+\beta_2)/p(A=i)$. 
Thus, there is a one-to-one correspondence between Huffman codes for partitions of $(\beta_1, \beta_2, \cdots, \beta_n)$ in which $\beta_1$ and $\beta_2$ are in the same group, and Huffman codes for partitions of $(\beta_1+\beta_2, \beta_3, \cdots, \beta_n)$. Moreover, from \eqref{eqnLXA2},  $L(X|A)$ for a
partition of $(\beta_1, \beta_2, \cdots, \beta_n)$ in which $\beta_1$ and $\beta_2$ are in the same group, equals $p(A=i)\times (\beta_1+\beta_2)/p(A=i)=\beta_1+\beta_2$ plus $L(X|A)$ for the corresponding partition of $(\beta_1+\beta_2, \beta_3, \cdots, \beta_n)$. Since $\beta_1+\beta_2$ is a constant that does not depend on the choice of partitions, it suffices to proceed by minimizing $L(X|A)$ over partitions of $(\beta_1+\beta_2, \beta_3, \cdots, \beta_n)$. This completes the proof. By applying the same argument iteratively $(n-k)$ times, we obtain a list of size $k$, for which the optimal partitioning is to place each element in a separate subset.

\end{proof}

\section{Number Partition Problem}\label{s3}
Let $S=(\beta_1, \beta_2, \cdots, \beta_n)$ be an arbitrary list of $n$ positive numbers. We can think of $S$ as being an unnormalized probability distribution. We can always normalize elements of $S$, and this scaling will not affect the balanced partitioning problem.
Given a list $S=(\beta_1, \beta_2, \cdots, \beta_n)$, we define a random variable $X$ with the alphabet set
 $\mathcal{X}=\{1,2,\cdots, n\}$ such that 
$\mathbb{P}[X=i]={\beta_i}/{M}$
where $M=\sum_{i}\beta_i$. Let $\mathcal{A}$ be a set of size $k$. Then, an $(n,k)$-partition function is a mapping $f:\mathcal{X} \to \mathcal{A}$. Here $\mathcal{A}=\{1,2,\cdots, k\}$. This partitions $\mathcal{X}$ into $k$ sets $f^{-1}(a)$ for $a\in \mathcal{A}$.  Let $A=f(X)$. Then, the marginal distribution of $A$ is proportional to the subset-sum values. Thus, each partition corresponds to a function
$
f:\mathcal{X}\rightarrow \mathcal{A}
$
where $\mathcal{A}$ is a set of size $k$ and the distribution on $X$ is specified by the given list of numbers.

\subsection{Compression objective function}

The Entropic objective function aims to maximize $H(A)$ where $A=f(X)$ for some function $f:\mathcal{X}\rightarrow\{1,2,\cdots, k\}$.
Observe that
$
H(X,A)=H(X)+H(A \vert X) =H(A)+H(X \vert A)$.
Since $H(A \vert X)=0$, we have $H(A)= H(X)-H(X \vert A)$. Since $H(X)$ does not depend on the choice of partition function, we can minimize $H(X \vert A)$ instead of maximizing $H(A)$. 
The conditional entropy $H(X \vert A)$ can be understood as the average uncertainty remaining in $X$ when $A$ is revealed. Moreover, $H(X \vert A)$ approximates the average number of bits required to compress the source $X$ when $A$ is revealed. Thus, we can define a new \emph{compression objective function} of minimizing $L(X|A)$ as defined in \eqref{eqnLXA2} as our new objective function for balanced partitioning. For this objective function, Algorithm \ref{algorithm} yields the optimal solution, which is actually very efficient. We believe that a similar procedure can be applied more broadly to combinatorial optimization problems such as those arising in combinatorial discrepancy by using information-theoretic concepts and metrics to redefine their objective functions. This approach may lead to new connections between information theory, computer science, and combinatorics. 

The authors in~\cite{rvh} show that Algorithm~\ref{algorithm}
computes an approximate solution to the entropy maximization problem over
partitions ($\max H(A)$), with a maximum deviation from the optimal solution bounded by
$\delta \approx 0.086$.  Consequently, the difference between the
optimal value of $H(A)$ and the entropy achieved by the optimal solution of
$L(X |A)$ is at most $\delta \approx 0.086$. Notably, this bound is
independent of both $k$ and $n$.

\subsection{Principle of optimality} 
A property of the Min-Difference objective function is that in each optimal $k$-way partition if the numbers in any $k-1$ subsets are optimally partitioned, the new partition is also optimal (principle of optimality) \cite{korf2009}.  This property underlies the recursive algorithms of \cite{korf2009}.  A different and a kind of more general principle of optimality (called recursive principle of optimality in \cite{korf2011}) is valid for Min-Max and Max-Min objective functions \cite{korf2010,korf2011}. It says that for any optimal $k$-way partition with $k$ subsets and any $k_1+k_2=k$, combining any optimal $k_1$-way partition of the numbers in $k_1$ subsets and any optimal $k_2$-way partition of the numbers in the other $k_2$ subsets results in an optimal partition for the main set \cite{korf2011}.  
In \cite{moffit2013}, the authors develop a \textit{principle of weakest-link optimality} for minimizing the largest subset sum. In \cite{korf2014}, the authors incorporate the ideas of \cite{korf2011,moffit2013} and \cite{korfethan2013} and develop an algorithm that is similar to \cite{moffit2013} in the sense of weakest-link optimality. See \cite{ethan18} for a review.

Next, we prove that the entropic objective function has a principle of optimality property similar to the one in \cite{korf2011} (which is the basis of algorithms in \cite{korf2011}).
\begin{theorem}
\label{pop} Take a random variable $X$ with alphabet set $\mathcal{X}$ and an optimal $(n,k)$-partition function $f:\mathcal{X}\to \mathcal{A}$ (for entropy objective function). Let  $\mathcal{A}_1$, $\mathcal{A}_2$ be an arbitrary partition of $\mathcal{A}$ into two sets.  Define a partition of $\mathcal{X}$ into $\mathcal{X}_1$ and $\mathcal{X}_2$ by 
$\mathcal{X}_i=f^{-1}(\mathcal{A}_i)$. 
Define a random variable $X_i$ with the alphabet set $\mathcal{X}_i$ whose distribution equals the conditional distribution of $X$ given $A\in \mathcal{A}_i$. Set  $k_i=\vert \mathcal{A}_i \vert,\
n_i=\vert  \mathcal{X}_i \vert, \  i \in \{1,2 \}$. Let  $f_1:\mathcal{X}_1\rightarrow \mathcal{A}_1$ be an arbitrary  optimal $(n_1,k_1)$-partition function of $X_1$, and $f_2:\mathcal{X}_2\rightarrow \mathcal{A}_2$ be an arbitrary optimal $(n_2,k_2)$-partition function of $X_2$. Then, the following function is an optimal $(n,k)$-partition function for $X$:
$$
  f_c(x) = \begin{cases}
  f_1(x) & x \in \mathcal{X}_1\\
  f_2(x) & x \in \mathcal{X}_2
  \end{cases}$$
    
\end{theorem}

The following known result is the key to prove Theorem~\ref{pop}.
\begin{lemma}[Grouping Axiom of Entropy]
\label{group}
For any probability vector $\boldsymbol{p}=(p_1,p_2,\ldots,p_k)$ and $1\leq r\leq k-1$,
\begin{align}
H(p_1,p_2,\ldots, p_k)=&H\left(\sum_{i=1}^{r}p_i, \sum_{i=r+1}^{k}p_i \right)\nonumber \\
&+\left( \sum_{i=1}^{r}p_i \right) H\left(\frac{p_1}{\sum_{i=1}^{r}p_i},\ldots,\frac{p_r}{\sum_{i=1}^{r}p_i} \right) \nonumber \\
&+\left( \sum_{i=r+1}^{k}p_i \right) H\left(\frac{p_{r+1}}{\sum_{i=r+1}^{k}p_{i}},\ldots,\frac{p_k}{\sum_{i=r+1}^{k}p_i} \right).
\end{align}
\end{lemma}

\begin{proof}[Proof of Theorem~\ref{pop}]
Let $f(X)=A$ be an optimal $(n,k)$-partition function of $X$.
Suppose that $f_1(X_1)=A'_1$ is an arbitrary optimal partition function for $X_1$. Thus, $H(A'_1) \geq H(f(X_1))$ by definition. On the other hand, using Lemma~\ref{group}, we have $H(A'_1) \leq H(f(X_1))$, because otherwise 
combining $f_1$ and $f\left\vert_{\mathcal{X}_2} \right.$ results in a $(n,k)$-partition function $f'(X)=A'$ such that $H(A')> H(A)$. That is a contradiction with the optimality of $f$.
A similar argument is true for $X_2$. Hence combining any optimal  $(n_1,k_1)$-partition function of $X_1$ with $(n_2,k_2)$-partition function of $X_2$ must yield an optimal $(n,k)$-partition function of $X$.    

\end{proof}
\begin{remark}
    One can directly check that the \emph{product objective function} as defined in \eqref{prod-objf} has a principle of optimality too.
\end{remark}

\subsection{$\alpha$-divergence objective function}

Considering the definition of $D_\alpha$ in Definition \ref{defR}, observe that the Min-Max objective function is equivalent to
\[
\min_{f:\mathcal{X}\rightarrow \mathcal{A}}D_{\infty}(q_A\|u_A)
\]
 where $A=f(X)\sim q_A$ and $u_A$ is the uniform distribution on the alphabet of $\mathcal{A}$. Similarly, the entropic objective function is equivalent to minimizing (over all functions $f:\mathcal{X}\rightarrow \mathcal{A}$) of  $D(q_A\|u_A)$ where $D$ is the KL divergence.

Now, instead of $D_{\infty}(q_A\|u_A)$, consider $D_{\infty}(u_A\|q_A)$. Note that the Max-Min objective function can be  expressed as 
\[
\min_{f:\mathcal{X}\rightarrow \mathcal{A}}D_{\infty}(u_A\|q_A)
\]
as it is equivalent to maximizing the smallest subset sum objective function. Next, if instead of minimizing $D_{\infty}(u_A\|q_A)$ we consider minimizing $D(u_A\|q_A)$, we obtain the \emph{product objective function}. These observations suggest that
\begin{align}
\label{entopt1}
\argmin_{\small \small{f:\mathcal{X}\rightarrow \mathcal{A} }} D_\alpha (u_{A} \| q_{A}). 
\end{align}  
 for some arbitrary $\alpha>0$ might be a suitable criterion to find a uniform partition. We call minimizing $D_{\alpha}(u_A\|q_A)$  the \emph{$\alpha$-divergence } objective function. 
 
 In \cite{sason}, $D_\alpha(q_A \| u_A)$ is considered while we consider $D_\alpha(u_A\|q_A)$. The following theorem uses similar ideas to derive lower and upper bounds for the
\emph{$\alpha$-divergence} objective function.

{\color{black}

\begin{theorem}
\label{main1}
Take a random variable $X$ on an alphabet set of size $n$ with probability mass vector $\boldsymbol{p}=(p_1,p_2,\ldots,p_n)$ ($p_i$'s are nonnegative and add up to one). Assume that $p_1\geq p_2\geq \cdots\geq p_n\geq 0$. Let $\mathcal{A}$ be a set of size $k$. Then
for all $2 \leq k<n$  and $\alpha>0$ we have
\begin{align}
\min_{\small \small{f:\mathcal{X}\rightarrow \mathcal{A} }} D_\alpha (u_{A} \| q_{A}) \in \left[D_\alpha(u_{A} \| q_{\hat{A}_k}),  D_\alpha(u_{A} \| q_{\hat{A}_k})+g(\alpha)\right],\label{boundalphdiv} 
\end{align}  
where $A=f(X)\sim q_A$ and $u_A$ is the uniform distribution on $\mathcal{A}$. In addition, $q_{\hat{A}_k}$  is the uniform distribution on $\{ 1, \ldots, k \}$ if $p_1<1/k$; otherwise, if $p_1\geq 1/k$, 
\begin{align}
\label{rm1}
 q_{\hat{A}_k}(i)=\begin{cases}
 p_i & i=1,\ldots,i^* \\
 \dfrac{1}{k-i^*} \sum_{j=i^*+1}^n p_j& i=i^*+1,\dots,k,
 \end{cases}
\end{align}
where $i^*$ is the maximum $i$ such that $p_i\geq \dfrac{1}{k-i}{ \sum_{j=i+1}^n p_j}$. Also,
\begin{align}
\label{ggg}
g(\alpha)=\begin{cases}
\log \left(\dfrac{\alpha-1}{2^{\alpha}-2}\right)-\dfrac{\alpha}{\alpha-1}\log \left(\dfrac{{\alpha}}{2^\alpha-1}\right) & \alpha \neq 1 \\
\log \left(\dfrac{2}{e\ln2}\right) \approx 0.08607& \alpha = 1.
 \end{cases}
\end{align}
Moreover, Algorithm \ref{algorithm} yields a value in the interval in \eqref{boundalphdiv} so it is within $g(\alpha)$ of the optimal solution.
\end{theorem}

\begin{proof}
\label{pnew}
The proof follows similar steps as in \cite{rvh}. Let $\hat{\boldsymbol{q}}=(\hat{q}_1,\hat{q}_2,\ldots,\hat{q}_k)$ be the vector presenting the probability mass function $q_{\hat{A}_k}(.)$, and
$\boldsymbol{P}_k$ for every $k\geq 2$ be the set of all probability mass vectors of size $k$. First we prove that
\begin{align}
\label{eqrrr}
    \min_{\small \small{f:\mathcal{X}\rightarrow \mathcal{A} }} D_\alpha (u_{A} \| q_{A}) \geq D_\alpha(u_{A} \| q_{\hat{A}_k}).
\end{align}
Using {\cite[Lemma~3]{rvh}}, the probability mass vector of each $(n,k)$-partition function $f(X)$ majorizes $\boldsymbol{p}$. On the other hand, \cite[Lemma~4]{rvh} shows that every probability mass vector $\boldsymbol{q} \in \boldsymbol{P}_k$ which majorizes $\boldsymbol{p}$  majorizes $\hat{\boldsymbol{q}}$ too. Therefore, Schur-convexity of the R\'enyi divergence yields \eqref{eqrrr}.

Algorithm~\ref{algorithm} is similar to the Huffman algorithm, but 
terminates after $n-k$ steps, when the size of the output list is $k$ (all
merged symbols belong to the same group in the resulting partition). Let $\boldsymbol{h}=(h_1,h_2,\ldots,h_k)$ be the output of the algorithm. The components $h_i$ represent the weights of the groups in the resulting
partition, sorted in decreasing order; each $h_i$ is obtained by merging a
subset of the original probabilities $p_j$.
Thus, we can let $A^H=f^H(X)$, where $f^H$ is induced by the
partition produced by Algorithm~\ref{algorithm} when the input is the
probability mass vector of $X$, $\boldsymbol{p}=(p_1,p_2,\ldots,p_n)$, and the
integer $k$. Here, 
$A^H=f^H(X)$, has pmf $\boldsymbol{h}=(h_1,h_2,\ldots,h_k)$.

To complete the proof of \eqref{boundalphdiv}, it suffices to show that 
\begin{align}
  D_\alpha \bigl(u_{A} \,\|\, q_{{A}^H}\bigr)
  \le
  D_\alpha \bigl(u_{A} \,\|\, q_{\hat{A}_k}\bigr)
  + g(\alpha),
\end{align}
where $q_{{A}^H}(\cdot)$ denotes the output of Algorithm~\ref{algorithm}
with the vector $\boldsymbol{p}$ as input, and $g(\alpha)$ is defined in
\eqref{ggg}.

Let $i^*$ be the largest index such that for each $j\leq i^*$ we have $h_j = p_j$ (i.e., $h_j$ is not constructed by merging steps in the algorithm). If $h_1 \neq p_1$, $i^*$ is set to be $0$. Let
\begin{align}
    \bar{\boldsymbol{q}}=\left(  h_{i^*+1}/S,\ldots, h_{k}/S  \right),
\end{align}
where $$S=\sum_{j=i^*+1}^k h_j $$ is the sum of the $k-i^*$ smallest masses of $\boldsymbol{h}$. It is shown in the proof of \cite[Lemma 5]{rvh} that 
\begin{align}
    \dfrac{\bar{q}_{\text{max}}}{\bar{q}_\text{min}}=\dfrac{h_{i^* +1}}{h_k} \leq 2.
\end{align}
Using Lemma~\ref{lemma-w} with $\rho =2$, we have
\begin{align}
\label{newb}
D_\alpha (\boldsymbol{u}_{k-i^*} \| \bar{\boldsymbol{q}}) \leq g(\alpha),
\end{align}
where $\boldsymbol{u}_{k-i^*}$ is the uniform mass vector of size $k-i^*$.
For $\alpha \in (0,1) \cup (1, \infty)$, we have
\begin{align}
 D_\alpha  ( \boldsymbol{u}_{k} \| {\boldsymbol{h}}) 
  &= \dfrac{1}{\alpha -1} \log \sum_{j=1}^{k} {\dfrac{1}{k^{\alpha}}} h_j^{1-\alpha}\\ \label{nsp}
 &= \dfrac{\alpha}{\alpha-1} \log \dfrac{1}{k}+ \dfrac{1}{\alpha-1} \log \sum_{j=1}^{k} h_j^{1-\alpha}.
 \end{align}
Now, continue with the second part of \eqref{nsp}
\begin{align}
&\dfrac{1}{\alpha-1}  \log \sum_{j=1}^{k} h_j^{1-\alpha} \\
 & \  = \dfrac{1}{\alpha-1} \log \left( \sum_{j=1}^{i^*} h_j^{1-\alpha} + \sum_{j=i^*+1}^{k} h_j^{1-\alpha} \right) \\ 
 & \ = \dfrac{1}{\alpha-1} \log \left( \sum_{j=1}^{i^*} h_j^{1-\alpha} +S^{1-\alpha} 2^{\log \sum_{j=1}^{k-i^*} \bar{q}_j^{1-\alpha}} \right)\\ \label{nd1}
  & \ \leq \dfrac{1}{\alpha-1} \log \left( \sum_{j=1}^{i^*} h_j^{1-\alpha}+S^{1-\alpha}(k-i^*)^\alpha 2^{ (\alpha-1)g(\alpha)} \right),
 \end{align}
where \eqref{nd1} follows from \eqref{newb}. Let
\begin{align}
\label{nqstar}
 q^*_j := \begin{cases}
 h_j & j=1,\ldots,i^*, \\
 \dfrac{S}{k-i^*} & j=i^*+1,\dots,k
 \end{cases}
\end{align}
be the components of the probability mass vector $\boldsymbol{q}^*$ of size $k$. From \eqref{nd1} and \eqref{nqstar} we have
\begin{align}
&\dfrac{1}{\alpha-1}  \log \sum_{j=1}^{k} h_j^{1-\alpha} \\    & \ \leq \dfrac{1}{\alpha-1} \log \left( \sum_{j=1}^{i^*} {q}_j^{*1-\alpha}+\sum_{j=i^*+1}^{k} q_j^{*1-\alpha}  2^{ (\alpha-1)g(\alpha)} \right) \\
& \ = \dfrac{1}{\alpha-1} \log \left( \sum_{j=1}^{k} {q}_j^{*1-\alpha}+\sum_{j=i^*+1}^{k} q_j^{*1-\alpha}  \left( 2^{ (\alpha-1)g(\alpha)}-1 \right) \right)\\ \nonumber
& \ = \dfrac{1}{\alpha-1} \log \sum_{j=1}^{k} {q}_j^{*1-\alpha}\\\label{neq1}
& \ \ \ \ + \dfrac{1}{\alpha-1} \log \left( 1+ w \left( 2^{ (\alpha-1)g(\alpha)}-1 \right) \right),
\end{align}
where 
\begin{align}
\label{neww}
    w=\dfrac{\sum_{j=i^*+1}^{k} q_j^{*1-\alpha}}{\sum_{j=1}^{k} q_j^{*1-\alpha}} \in [0,1].
\end{align}
For $\alpha > 0$, we know that $g(\alpha)>0$ ($g(\alpha)$ is $c_\alpha^\infty (2)$ in \cite[Lemma~4]{sason}), therefore, by using \eqref{neww}  for $\alpha \in (1, \infty)$ we have

\begin{align}
    & \dfrac{1}{\alpha-1} \log \left( 1+ w \left( 2^{ (\alpha-1)g(\alpha)}-1 \right) \right) \\
 & \ \leq  \dfrac{1}{\alpha-1} \log \left( 1+ \left( 2^{ (\alpha-1)g(\alpha)}-1 \right) \right)\\
  & \ \leq  \dfrac{1}{\alpha-1} \log \left( 2^{ (\alpha-1)g(\alpha)} \right) \\
 & \ \leq g(\alpha).
\end{align}
 Similarly, for $\alpha \in (0, 1)$, the following inequality holds.
\begin{align}
\label{nfinal}
    \dfrac{1}{\alpha-1} \log \left( 1+ w \left( 2^{ (\alpha-1)g(\alpha)}-1 \right) \right) \leq g(\alpha).
\end{align}
As a result, for $\alpha \in (0,1) \cup (1, \infty)$ we get
\begin{align} 
 &  D_\alpha  ( \boldsymbol{u}_{k} \| {\boldsymbol{h}}) \\ 
 & \ = \dfrac{\alpha}{\alpha-1} \log \dfrac{1}{k}+ \dfrac{1}{\alpha-1} \log \sum_{j=1}^{k} h_j^{1-\alpha}\\ \label{neq2}
  & \ \leq \dfrac{\alpha}{\alpha-1} \log \dfrac{1}{k}+ \dfrac{1}{\alpha-1} \log \sum_{j=1}^{k} {q}_j^{*1-\alpha} +g(\alpha) \\
  & \ = D_\alpha  ( \boldsymbol{u}_{k} \| {\boldsymbol{q}^*})+g(\alpha),
\end{align}
where \eqref{neq2} follows from \eqref{neq1} and \eqref{nfinal}. Since $g(\alpha)$ is continuous at $\alpha =1$,  the inequality
\begin{align}
  D_\alpha  ( \boldsymbol{u}_{k} \| {\boldsymbol{h}}) \leq   D_\alpha  ( \boldsymbol{u}_{k} \| {\boldsymbol{q}^*})+g(\alpha)
\end{align}
is also valid for $\alpha=1$, by taking the limit on both sides. Finally, it is shown in the proof of \cite[Lemma 5]{rvh} that  $\boldsymbol{q}^* \preceq \hat{\boldsymbol{q}}$.
Together with the Schur-convexity of $D_\alpha(\boldsymbol{u}_k \| \cdot)$, this
implies that
\begin{align}
     D_\alpha\!\left( \boldsymbol{u}_k \,\|\, \boldsymbol{h} \right)
     \le
     D_\alpha\!\left( \boldsymbol{u}_k \,\|\, \hat{\boldsymbol{q}} \right)
     + g(\alpha),
\end{align}
for all $\alpha>0$. 
\end{proof}

\section{Conclusion}
We solved a new source coding problem and related it to the number partitioning problem. In this paper, we introduce two new objective functions for a classical
computer science problem, the number partitioning problem: the compression
objective and the $\alpha$-divergence objective functions. We prove the
optimality of a variant of the Huffman algorithm with respect to the
compression objective. Additionally, we provide upper and lower bounds
for the $\alpha$-divergence objective function. A principle of optimality, which
plays a crucial role in algorithm design, is established for both the
$1$-divergence objective and the entropy objective functions introduced in
previous work.

\bibliographystyle{IEEEtran}
\bibliography{IEEEabrv,bibtex}

\begin{thebibliography}{10}
\providecommand{\url}[1]{#1}
\csname url@samestyle\endcsname
\providecommand{\newblock}{\relax}
\providecommand{\bibinfo}[2]{#2}
\providecommand{\BIBentrySTDinterwordspacing}{\spaceskip=0pt\relax}
\providecommand{\BIBentryALTinterwordstretchfactor}{4}
\providecommand{\BIBentryALTinterwordspacing}{\spaceskip=\fontdimen2\font plus
\BIBentryALTinterwordstretchfactor\fontdimen3\font minus
  \fontdimen4\font\relax}
\providecommand{\BIBforeignlanguage}[2]{{%
\expandafter\ifx\csname l@#1\endcsname\relax
\typeout{** WARNING: IEEEtran.bst: No hyphenation pattern has been}%
\typeout{** loaded for the language `#1'. Using the pattern for}%
\typeout{** the default language instead.}%
\else
\language=\csname l@#1\endcsname
\fi
#2}}
\providecommand{\BIBdecl}{\relax}
\BIBdecl

\bibitem{rvh}
F.~Cicalese, L.~Gargano, and U.~Vaccaro, ``Bounds on the entropy of a function
  of a random variable and their applications,'' \emph{IEEE Transactions on
  Information Theory}, vol.~64, no.~4, pp. 2220--2230, 2017.

\bibitem{garey79}
M.~R. Garey and D.~S. Johnson, \emph{Computers and intractability}.\hskip 1em
  plus 0.5em minus 0.4em\relax freeman San Francisco, 1979, vol. 174.

\bibitem{ggreedy69}
R.~L. Graham, ``Bounds on multiprocessing timing anomalies,'' \emph{SIAM
  journal on Applied Mathematics}, vol.~17, no.~2, pp. 416--429, 1969.

\bibitem{walsh}
T.~Walsh, ``Where are the really hard manipulation problems? the manipulation
  phase transition,'' in \emph{Proc. IJCAI}, vol.~9, 2009.

\bibitem{worst2007}
W.~Michiels, J.~Korst, E.~Aarts, and J.~Van~Leeuwen, ``Performance ratios of
  the karmarkar-karp differencing method,'' \emph{Journal of combinatorial
  optimization}, vol.~13, no.~1, pp. 19--32, 2007.

\bibitem{korf98}
R.~E. Korf, ``A complete anytime algorithm for number partitioning,''
  \emph{Artificial Intelligence}, vol. 106, no.~2, pp. 181--203, 1998.

\bibitem{ethan18}
E.~L. Schreiber, R.~E. Korf, and M.~D. Moffitt, ``Optimal multi-way number
  partitioning,'' \emph{J. ACM}, vol.~65, no.~4, Jul. 2018.

\bibitem{korf2010}
R.~E. Korf, ``Objective functions for multi-way number partitioning,'' in
  \emph{Third Annual Symposium on Combinatorial Search}, 2010.

\bibitem{karp72}
R.~M. Karp, ``Reducibility among combinatorial problems,'' in \emph{Complexity
  of computer computations}, 1972, pp. 85--103.

\bibitem{korf2009}
R.~E. Korf, ``Multi-way number partitioning,'' in \emph{Twenty-First
  International Joint Conference on Artificial Intelligence}, 2009.

\bibitem{korf2014}
R.~E. Korf, E.~L. Schreiber, and M.~D. Moffitt, ``Optimal sequential multi-way
  number partitioning.'' in \emph{ISAIM}, 2014.

\bibitem{moffit2013}
M.~D. Moffitt, ``Search strategies for optimal multi-way number partitioning,''
  in \emph{Twenty-Third International Joint Conference on Artificial
  Intelligence}, 2013.

\bibitem{korf95}
R.~E. Korf, ``From approximate to optimal solutions: A case study of number
  partitioning,'' in \emph{IJCAI}, 1995, pp. 266--272.

\bibitem{karp82}
N.~Karmarkar and R.~M. Karp, \emph{The differencing method of set
  partitioning}.\hskip 1em plus 0.5em minus 0.4em\relax University of
  California Berkeley, 1982.

\bibitem{korf2011}
R.~E. Korf, ``A hybrid recursive multi-way number partitioning algorithm,'' in
  \emph{Twenty-Second International Joint Conference on Artificial
  Intelligence}, 2011.

\bibitem{sason}
I.~Sason, ``Tight bounds on the r{\'e}nyi entropy via majorization with
  applications to guessing and compression,'' \emph{Entropy}, vol.~20, no.~12,
  p. 896, 2018.

\bibitem{renyi}
A.~R{\'e}nyi \emph{et~al.}, ``On measures of entropy and information,'' in
  \emph{Proceedings of the Fourth Berkeley Symposium on Mathematical Statistics
  and Probability, Volume 1: Contributions to the Theory of Statistics}.\hskip
  1em plus 0.5em minus 0.4em\relax The Regents of the University of California,
  1961.

\bibitem{sasonnew}
I.~Sason, ``On data-processing and majorization inequalities for f-divergences
  with applications,'' \emph{Entropy}, vol.~21, no.~10, p. 1022, 2019.

\bibitem{huffman}
D.~A. Huffman, ``A method for the construction of minimum-redundancy codes,''
  \emph{Proceedings of the IRE}, vol.~40, no.~9, pp. 1098--1101, 1952.

\bibitem{korfethan2013}
R.~E. Korf and E.~L. Schreiber, ``Optimally scheduling small numbers of
  identical parallel machines,'' in \emph{Twenty-Third International Conference
  on Automated Planning and Scheduling}, 2013.

\end{thebibliography}

\end{document}